\documentclass[letterpaper, 10 pt, conference]{IEEEtran}  
\usepackage{graphicx}      
\usepackage{graphics} 
\usepackage{epsfig} 
\usepackage{amsmath} 
\usepackage{algorithm,algorithmic}
\usepackage{balance}
\usepackage{amsmath,amsfonts,amssymb,amscd}
\usepackage{capt-of}
\usepackage{bbm}

\usepackage{bm}

\usepackage{xcolor}
\usepackage{verbatim}
\usepackage{mathtools}

\newtheorem{remark}{\bfseries Remark}
\newtheorem{theorem}{\bfseries Theorem}

\newtheorem{assumption}{\bfseries Assumption}

\newenvironment{varalgorithm}[1]
  {\algorithm\renewcommand{\thealgorithm}{#1}}
  {\endalgorithm}

\newenvironment{list4}{
	\begin{list}{$\bullet$}{%
			\setlength{\itemsep}{0.05cm}
			\setlength{\labelsep}{0.2cm}
			\setlength{\labelwidth}{0.3cm}
			\setlength{\parsep}{0in} 
			\setlength{\parskip}{0in}
			\setlength{\topsep}{0in} 
			\setlength{\partopsep}{0in}
			\setlength{\leftmargin}{0.16in}}}
	{\end{list}}

\newenvironment{list4a}{
	\begin{list}{$\bullet$}{%
			\setlength{\itemsep}{0.05cm}
			\setlength{\labelsep}{0.2cm}
			\setlength{\labelwidth}{0.3cm}
			\setlength{\parsep}{0in} 
			\setlength{\parskip}{0in}
			\setlength{\topsep}{0in} 
			\setlength{\partopsep}{0in}
			\setlength{\leftmargin}{0.16in}}}
	{\end{list}}

\IEEEoverridecommandlockouts                              

\title{\LARGE \bf
Distributed Quantized Average Consensus in Open Multi-Agent Systems with Dynamic Communication Links
}

\author{Jiaqi Hu, Karl H.~Johansson, and Apostolos~I.~Rikos
\thanks{Jiaqi Hu and Apostolos~I.~Rikos are with the AI Thrust, Information Hub, The Hong Kong University of Science and Technology (Guangzhou), Guangzhou, China. 
Apostolos~I.~Rikos is also affiliated with the Department of Computer Science and Engineering, The Hong Kong University of Science and Technology, Clear Water Bay, Hong Kong. 
E-mails: {\tt~jhu021@connect.hkust-gz.edu.cn; apostolosr@hkust-gz.edu.cn}}
\thanks{Karl H.~Johansson is with the Division of Decision and Control Systems, KTH Royal Institute of Technology, SE-100 44 Stockholm, Sweden. 
    He is also affiliated with Digital Futures, SE-100 44 Stockholm, Sweden. 
    E-mail:{\tt~kallej@kth.se}.
        }
}

\begin{document}

\maketitle
\thispagestyle{empty}
\pagestyle{empty}

\begin{abstract}
In this paper, we focus on the distributed quantized average consensus problem in open multi-agent systems consisting of communication links that change dynamically over time. 
Open multi-agent systems exhibiting the aforementioned characteristic are referred to as \textit{open dynamic multi-agent systems} in this work. 
We present a distributed algorithm that enables active nodes in the open dynamic multi-agent system to calculate the quantized average of their initial states. 
Our algorithm consists of the following advantages: (i) ensures efficient communication by enabling nodes to exchange quantized valued messages, and (ii) exhibits finite time convergence to the desired solution. 
We establish the correctness of our algorithm and we present necessary and sufficient topological conditions for it to successfully solve the quantized average consensus problem in an open dynamic multi-agent system. 
Finally, we illustrate the performance of our algorithm with numerical simulations. 

\end{abstract}

\section{Introduction}\label{sec:intro}
In recent years, multi-agent systems (MAS) have gathered significant attention due to the growing interest of the scientific community in control and coordination algorithms. 
MAS consist of a set of nodes (or agents) that interact and collaborate to achieve a common objective. 
Examples include social networks, networks of devices such as sensors or computing entities, or groups of vehicles or robots \cite{2007:olfati-saber_consensus}. 

One problem of particular importance in distributed control is the distributed average consensus problem. 
In this problem each node in the network starts with an initial state. 
The goal is to calculate the average of these initial states through local interactions and information exchange among neighboring nodes (see \cite{2018:BOOK} and references therein).  
The distributed average consensus problem has received extensive attention with various proposed algorithms focusing on different scenarios, such as: real-valued communication among nodes in \cite{SEYBOTH:2013}, quantized communication among nodes in \cite{2021:Rikos_Hadj_Accumul_TAC}, and operation over unreliable networks in \cite{chrisTAC:2016}. 

Recently, various works have focused on developing algorithms for MAS that can handle a more complex scenario: nodes entering and leaving the network. 
This emerging category of systems is known as open multi-agent systems (OMAS) \cite{Vizuete2024}. 
Specifically, trust and reputation models were introduced in \cite{huynh2006integrated}, while \cite{2017_Hendrickx_Martin_CDC} developed a framework based on pairwise average gossip algorithms. In \cite{golpayegani2019using} the authors proposed a model leveraging social reasoning to improve collaboration among nodes. 
The behavior of average pairwise gossip algorithms in OMAS was investigated in \cite{2016_Hendrickx_Allerton}, and \cite{2021_Franceschelli_Frasca_TAC} provided an analysis of OMAS dynamical properties. 
Further advancements include distributed mode computation, as explored in \cite{2022_Dashti_Mauro_IEEELCSS}, and a framework for heterogeneous pairwise interactions among nodes presented in \cite{2024_Oliva_Scala_TAC_Open}.
Shifting focus, research has expanded into distributed optimization for OMAS in \cite{2020_Hendrickx_Rabbat_CDC, 2023_Hayashi_TAC}, while \cite{2023_Nakamura_Inuiguchi_IEEECSS} explored adversarial multi-armed bandits in the same context. 
As aforementioned works assumed undirected communication between nodes, more recent works relaxed this assumption, focusing on OMAS that exhibit directed communication. 
Advancements in this direction include the distributed averaging framework based on running-sums in \cite{2024:CDC_Hadjic_Garcia}, and the averaging approach utilizing acknowledgment feedback signals in \cite{2024:CDC_Themis_Open}. 

It is important to note here that current approaches in the literature focus on algorithms that require nodes to exchange real-valued messages. 
This imposes high bandwidth requirements and compromises their resource efficiency, limiting their applicability in real-world scenarios (where communication links may have limited capacity). 
Additionally, current approaches assume constant and reliable connectivity among active nodes (i.e., nodes currently participating in the OMAS). 
Possible changes in the OMAS's communication links (e.g., due to mobility of active components or environmental interference) are not adequately addressed, thus restricting even further the practical implementation of current approaches. 
To the authors' knowledge, no existing approach addresses the challenges of limited bandwidth and dynamic communication in open multi-agent systems. 
As a result, developing communication-efficient distributed algorithms able to operate over open multi-agent systems consisting of dynamic communication links remains an open problem in the literature. 


\noindent 
\textbf{Main Contributions.} 
Motivated by the aforementioned gap in the literature, we present a distributed average consensus algorithm designed to operate in open multi-agent systems with dynamic communication links. 
Our algorithm is the first to address the challenges of limited bandwidth and dynamic connectivity in open multi-agent systems with directed communication. 
Our main contributions are the following: 
\begin{list4}
    \item We introduce a novel distributed algorithm that enables nodes to almost surely reach quantized average consensus in finite time while operating in an open multi-agent system with dynamic communication links. 
    Unlike existing approaches in the literature, our algorithm enables nodes to (i) operate in a communication efficient manner by exchanging quantized valued messages, and (ii) operate in the presence of dynamically-changing directed communication links in the network (see Algorithm~\ref{algorithm1}). 
    \item We analyze the operation of our proposed algorithm and establish its finite time convergence. 
    Additionally, we provide necessary and sufficient topological conditions that ensure our algorithm's correctness despite the presence of dynamic directed communication links (see Theorem~\ref{main_convergence_condition_theorem}). 
\end{list4}
The algorithm in \cite{2022:Rikos_Hadj_Johan} serves as the underlying framework for designing our proposed algorithm. 
Note that \cite{2022:Rikos_Hadj_Johan} is not designed to operate in open multi-agent systems, as it cannot handle arrivals and departures of nodes. 
Our current work enhances its functionality by introducing modes that define the operation of each node depending on its status (i.e., if a node is participating, departing, or arriving in the system). 


%
%
%
%
\section{NOTATION AND BACKGROUND}\label{sec:preliminaries}

\noindent
\textbf{Notation.}
The sets of real, rational, integer and natural numbers are denoted by $\mathbb{R}$, $\mathbb{Q}$, $\mathbb{Z}$, and $\mathbb{N}$, respectively. 
The symbol $\mathbb{Z}_+$ denotes the set of nonnegative integers, while $\mathbb{Z}_0$ denotes the set of natural numbers that includes zero. 
For two sets $A$ and $B$, $A \cap B$ denotes set intersection, 
$A \cup B$ represents set union, 
and $A \setminus B$ denotes set difference. 
For any real number $a \in \mathbb{R}$, the floor $\lfloor a \rfloor$ denotes the greatest integer less than or equal to $a$ while the ceiling $\lceil a \rceil$ denotes the least integer greater than or equal to $a$. 
\vspace{.2cm}

Let us now introduce two key network concepts (i) open networks, and (ii) dynamic networks.
By integrating these concepts, we present the network model that we consider in our current work. 
In our work we focus in open dynamic networks (i.e., networks in which nodes enter and leave while the communication links are time-varying).  

\subsection{Open Networks}\label{subsec_open_networks}

In open networks the communication topology is directed and open. 
Directed topology means the flow of information (or interaction) between nodes follows specific directions (i.e., if node $v_j$ can transmit information to node $v_l$, it does not necessarily imply that node $v_l$ can transmit information to node $v_j$). 
Open topology means that the network allows nodes to have the freedom to enter or leave the network at their discretion. 
Open networks consist of $n$ nodes (where $n \geq 2$) communicating only with their immediate neighbors at any time step $k$. 
The finite set of all $n$ nodes \textit{potentially participating} in the open network is captured by $\mathcal{V}^\prime = \{v_1, v_2, . . ., v_n\}$. 
However, since nodes enter or leave the network, at each time step $k$ only a subset of nodes $\mathcal{V}^o[k] \subseteq \mathcal{V}^\prime$ is considered active. 
Active nodes at time step $k$ are nodes that are present in the network at time step $k$ and capable of sending or receiving information. 
Conversely, inactive nodes at time step $k$ are nodes that are not present in the network at time step $k$, and are incapable of sending or receiving information. 
We assume that interactions occur only between nodes that are active. 
Following this, the communication topology of the active nodes is modeled as an open digraph $\mathcal{G}^o_d[k]=(\mathcal{V}^o[k], \mathcal{E}^o[k])$, with $\mathcal{V}^o[k] \subseteq \mathcal{V}^\prime$ denoting the set of active nodes at time step $k$, and $\mathcal{E}^o[k] \subseteq \mathcal{V}^o[k] \times \mathcal{V}^o[k]$ denoting the set of edges between active nodes at time step $k$. 
The cardinality of active nodes at time step $k$ is denoted as $n^o[k] = | \mathcal{V}^o[k] |$. 
A directed edge from node $v_j$ to node $v_l$ is denoted by $m^o_{lj} \triangleq (v_l, v_j) \in \mathcal{E}^o[k]$. 
It captures the fact that node $v_j$ can transmit information to node $v_l$ (but not necessarily the other way around). 
The set of all possible edges in the network is denoted as $\mathcal{E}^o = \cup_{k=0}^{\infty} \mathcal{E}^o[k]$. 
The subset of nodes that can directly transmit information to node $v_j$ is called the set of in-neighbors of $v_j$ and is denoted as $\mathcal{N}^{-,o}_j[k] = \{v_i \in \mathcal{V}^o[k] | (v_j, v_i) \in \mathcal{E}^o[k]\}$. 
The cardinality of $\mathcal{N}^{-,o}_j[k]$ is called the in-degree of $v_j$ denoted as $D^{-,o}_j[k] = | \mathcal{N}^{-,o}_j[k] |$. 
The subset of nodes that can directly receive information from node $v_j$ is called the set of out-neighbors of $v_j$ and is denoted as $\mathcal{N}^{+,o}_j[k] = \{v_l \in \mathcal{V}^o[k] | (v_l, v_j) \in \mathcal{E}^o[k]\}$. 
The cardinality of $\mathcal{N}^{+,o}_j[k]$ is called the out-degree of $v_j$ denoted as $D^{+,o}_j[k] = | \mathcal{N}^{+,o}_j[k] |$. 

At each time step $k$, every node $v_j$ belongs in one of the following three subsets (i.e., it operates according to one of the following three operating modes). 

\textbf{Remaining.} 
This subset is denoted by $\mathcal{R}^o[k]$. 
It comprises nodes that are active at both time steps $k$ and $k+1$. Specifically, $\mathcal{R}^o[k]$ represents the nodes that maintain their active status for time steps $k$ and $k+1$ (i.e., nodes that are active in the network at time step $k$, and are not departing in the next time step $k+1$). 
It is defined as 
\begin{equation}\label{remain_set_defn} 
    \mathcal{R}^o[k] = \mathcal{V}^o[k] \cap \mathcal{V}^o[k+1]. 
\end{equation}

\textbf{Arriving.} 
This subset is denoted by $\mathcal{A}^o[k]$. 
It comprises nodes that transition from inactive to active state at time step $k$. 
Specifically, $\mathcal{A}^o[k]$ consists of nodes that are inactive at time step $k$ but become active at time step $k+1$.
It is defined as 
\begin{equation}\label{arrive_set_defn}
    \mathcal{A}^o[k] = \mathcal{V}^o[k+1] \setminus \mathcal{V}^o[k]. 
\end{equation}

\textbf{Departing.} 
This subset is denoted by $\mathcal{D}^o[k]$. 
It comprises nodes that are active at time step $k$ but become inactive at time step $k+1$. 
Specifically, $\mathcal{D}^o[k]$ represents the nodes that exit the network at time step $k$. 
It is defined as 
\begin{equation}\label{depart_set_defn} 
\mathcal{D}^o[k] = \mathcal{V}^o[k] \setminus \mathcal{V}^o[k+1]. 
\end{equation}

Nodes in an open network operate according to the previously described modes, transitioning between states as follows. 
At each time step $k$, the inactive nodes can choose the arriving mode to become active. 
The active nodes can choose (i) the departing mode to leave the network and become inactive, or (ii) the remaining mode to remain active. 

\subsection{Dynamic Networks}\label{subsec_dynam_networks}

For introducing dynamic networks we borrow the following description from  \cite[Section~$2.1$]{2022:Rikos_Hadj_Johan}. 
In dynamic networks the communication topology is directed and dynamic (i.e., communication links among nodes change over time). 
The dynamically changing directed network can be captured by a sequence of directed graphs (digraphs), defined as $\mathcal{G}_d^d[k] = (\mathcal{V}^\prime, \mathcal{E}^d[k])$ (for $k=0, 1, ...$), where $\mathcal{V}^\prime =  \{v_1, v_2, \dots, v_n\}$ is the set of nodes and $\mathcal{E}^d[k] \subseteq \mathcal{V}^\prime \times \mathcal{V}^\prime$ is the set of edges at time step $k$. 
Note that several key definitions remain consistent with those presented in Section~\ref{subsec_open_networks}. 
These include (i) the cardinality of nodes denoted as $n^d = | \mathcal{V}^\prime |$, (ii) a directed edge from node $v_j$ to node $v_l$ represented as $m^d_{lj}$, (iii) the sets of in-neighbors and out-neighbors of a node $v_j$ denoted as $\mathcal{N}^{-,d}_j[k]$ and $\mathcal{N}^{+,d}_j[k]$, respectively, and (iv) the cardinalities of these neighbor sets $\mathcal{N}^{-,d}_j[k]$ and $\mathcal{N}^{+,d}_j[k]$, represented by $D^{-,d}_j[k]$ and $D^{+,d}_j[k]$, respectively. 
Given a dynamic digraph $\mathcal{G}_d^d[k] = (\mathcal{V}^\prime, \mathcal{E}^d[k])$ for $k = 1, 2, ..., m$, where $m \in \mathbb{N}$, its \textit{union digraph} is defined as $\mathcal{G}^{d, \{1, 2, ..., m\}}_d = (\mathcal{V}^\prime, \cup_{k = 1}^{m} \mathcal{E}^d[k])$. 
A dynamic digraph is \textit{jointly strongly connected} over the interval $k = 1, 2, ..., m$ if its corresponding union graph $\mathcal{G}^{d, \{1, 2, ..., m\}}_d$ forms a strongly connected digraph. 
A digraph is strongly connected if for each pair of nodes $v_j, v_l$, (where $v_j \neq v_l$) there exists a directed \textit{path} from $v_j$ to $v_l$. 
This means that in the union graph, for any two distinct nodes $v_j, v_l \in \mathcal{V}^\prime$, there exists a directed \textit{path} from $v_j$ to $v_l$. 

\subsection{Open Dynamic Networks}\label{subsec_open_dynam_networks}
In our current work, we explore networks characterized by a communication topology that is simultaneously directed, open, and dynamic. 
More specifically, we focus on networks in which at every time step $k$ (i)  the flow of information between nodes follows specific directions, (ii) nodes have the freedom to enter or leave the network, and (iii) communication links among nodes change over time. 
Note here that a key similarity between open dynamic networks and the previously discussed open networks is the separation of nodes in active and inactive. 
However, the critical difference between these two network types lies in the behavior of communication links among active nodes. 
In open networks, these links remain static over time. 
In contrast, in open dynamic networks communication links change dynamically even among nodes that are active in the network.

Open dynamic networks consist of $n$ nodes (where $n \geq 2$), and can be captured by a sequence of digraphs defined as $\mathcal{G}_d[k]=(\mathcal{V}[k], \mathcal{E}[k])$. 
The subset $\mathcal{V}[k] \subseteq \mathcal{V}^\prime$ denotes the set of active nodes at time step $k$. 
Note that the subsets of active and inactive nodes along with the set of all $n$ nodes potentially participating to the open network (captured by $\mathcal{V}^\prime = \{v_1, v_2, . . ., v_n\}$) were defined in Section~\ref{subsec_open_networks}. 
The subset $\mathcal{E}[k] \subseteq \mathcal{V}[k] \times \mathcal{V}[k]$ denotes the set of edges between active nodes at time step $k$. 
Note here that $\mathcal{V}[k] = \mathcal{V}^{o}[k]$, and $\mathcal{E}[k] = \mathcal{E}^{o}[k] \cap \mathcal{E}^{d}[k]$, during every time step $k$ (where $\mathcal{V}^{o}[k], \mathcal{E}^{o}[k]$ were defined in Section~\ref{subsec_open_networks}, and $ \mathcal{E}^{d}[k]$ was defined in Section~\ref{subsec_dynam_networks}). 
At every time step $k$, the cardinality of active nodes is denoted as $n[k] = | \mathcal{V}[k] |$. 
A directed edge from node $v_j$ to node $v_l$ is denoted by $m_{lj} \triangleq (v_l, v_j) \in \mathcal{E}[k]$, 
and the cardinality of edges at each time step $k$ is denoted as $m[k] = | \mathcal{E}[k] |$. 
The set of all possible edges in the network is denoted as $\mathcal{E} = \cup_{k=0}^{\infty} \mathcal{E}[k]$. 
The set of in-neighbors of node $v_j$ is denoted as $\mathcal{N}^{-}_j[k] = \{v_i \in \mathcal{V}[k] \ | \ (v_j, v_i) \in \mathcal{E}[k]\}$, and the set of out-neighbors as $\mathcal{N}^{+}_j[k] = \{v_l \in \mathcal{V}[k] \ | \ (v_l, v_j) \in \mathcal{E}[k]\}$. 
Note that $\mathcal{N}^{-}_j[k] = \mathcal{N}^{-,o}_j[k] \cap \mathcal{N}^{-,d}_j[k]$, and $\mathcal{N}^{+}_j[k] = \mathcal{N}^{+,o}_j[k] \cap \mathcal{N}^{+,d}_j[k]$ (where $\mathcal{N}^{-,o}_j[k]$, $\mathcal{N}^{+,o}_j[k]$ were defined in Section~\ref{subsec_open_networks}, and $\mathcal{N}^{-,d}_j[k]$, $\mathcal{N}^{+,d}_j[k]$ were defined in Section~\ref{subsec_dynam_networks}). 
The in-degree and out-degree of $v_j$ are denoted as $D^{-}_j[k] = | \mathcal{N}^{-}_j[k] |$, and $D^{+}_j[k] = | \mathcal{N}^{+}_j[k] |$, respectively. 

Since we have $\mathcal{V}[k] = \mathcal{V}^{o}[k]$ at every time step $k$, 
in open dynamic networks nodes are categorized into operational modes similar to those in open networks. 
This categorization was described in Section~\ref{subsec_open_networks}. 
Specifically, at each time step $k$, every node belongs to one of the three following subsets: 
Remaining (denoted as $\mathcal{R}[k] = \mathcal{R}^o[k]$), 
Arriving (denoted as $\mathcal{A}[k] = \mathcal{A}^o[k]$), or 
Departing (denoted as $\mathcal{D}[k] = \mathcal{D}^o[k]$). 
The definitions of the aforementioned three subsets are analogous to equations \eqref{remain_set_defn}, \eqref{arrive_set_defn}, and \eqref{depart_set_defn}, respectively. 
Nodes transition between these three subsets as in open networks.
The main difference in open dynamic networks (compared to open networks) lies in the behavior of communication links. 
The links between active nodes also change over time. 
This introduces additional challenges for analyzing the network structure and designing distributed algorithms that operate over it.

\section{Problem Formulation}\label{sec:probForm}

In this paper, we aim to address problem \textbf{P1} below. 



\textbf{P1.} 
Let us consider an open dynamic network $\mathcal{G}_d[k] = (\mathcal{V}[k], \mathcal{E}[k])$ (defined in Section~\ref{subsec_open_dynam_networks}). 
Each node $v_j \in \mathcal{V}$ has an initial quantized state $x_j$ (for simplicity we have $x_j \in \mathbb{Z}$). 
Let us define $q[k]$ as the real average of the initial states of active nodes at time step $k$, given by 
\begin{equation}\label{goal}
    q[k] = \frac{1}{n[k]}\sum\limits_{v_j \in \mathcal{V}[k]} x_j ,  
\end{equation}
where $n[k]$ is the number of active nodes, and $\mathcal{V}[k]$ is the set of active nodes, at time step $k$. 
Our goal is to develop a distributed algorithm that enables active nodes to compute in finite time a quantized state equal to either the floor or ceiling of the real average $q[k]$ of the initial states (as defined in~\eqref{goal}).
Specifically, we require that there exists $k_0$ so that for every active node $ v_j \in \mathcal{V}[k]$ we have 
\begin{equation}\label{alpha_q_no_oscill}
( q^s_j[k] \!= \!\lfloor q[k] \rfloor \ \ \text{for} \ \ k \geq k_0 ) \ \ \text{or} \ \ ( q^s_j[k] \!=\! \lceil q[k] \rceil \ \ \text{for} \ \ k \geq k_0).
\end{equation}
Additionally, for enhancing operational and resource efficiency, nodes are required to communicate by exchanging quantized valued messages.

    


%
%
%
%
\section{Quantized Averaging in Open Dynamic Networks} 
\label{sec:distr_algo}

In this section we present a distributed algorithm that addresses problem \textbf{P1} in Section~\ref{sec:probForm}. 
Our algorithm is detailed below as Algorithm~\ref{algorithm1}. 
Before we present its main functionalities, we establish the following set of assumptions, which are necessary for our subsequent development. 


\begin{assumption}\label{awareness_of_remaining_out_neighbors}
At every time step $k \geq 0$, each active node $v_j \in \mathcal{V}[k]$ knows the set of its remaining out-neighbors $\mathcal{N}^{+}_j[k] \cap \mathcal{R}[k]$ (i.e., out-neighbors that will stay active in the next time step). 
\end{assumption} 


\begin{assumption}\label{existence_stable_time_step} 
There exists a time step $k'$ such that 
\begin{equation}\label{stable_nodepart_eq}
    \mathcal{V}[k] = \mathcal{V}_{\mathcal{R}}, \ \forall k \geq k' , 
\end{equation} 
where $\mathcal{V}_{\mathcal{R}}$ is the set of remaining nodes for time steps $k \geq k'$ (i.e., the subset of active nodes in the open dynamic network stabilizes for time steps $k \geq k'$). 
\end{assumption}

\begin{assumption}[$T$-jointly strong connectivity]\label{strong_connectivity_stable_union_graph}
Let us consider an open dynamic network $\mathcal{G}_d[k] = (\mathcal{V}[k], \mathcal{E}[k])$. 
For $k \geq k'$, each $\mathcal{G}_d[k]$ takes a value among a finite set of instances $\{ \mathcal{G}_{d_1}$, $\mathcal{G}_{d_2}$, ..., $\mathcal{G}_{d_T} \}$, where $\mathcal{G}_{d_\theta} = (\mathcal{V}_{\mathcal{R}}, \mathcal{E}_{d_\theta})$ for some $\theta \in \{1, 2, ..., T\}$, and $T \in \mathbb{N}$ (note that $\mathcal{V}_{\mathcal{R}}$ was defined in Assumption~\ref{existence_stable_time_step}). 
Specifically, at each time step $k \geq k'$, we have $\mathcal{G}_d[k] = \mathcal{G}_{d_{\theta}}$ for some $\theta \in \{1, 2, ..., T\}$ with probability $p_\theta > 0$ where $\sum_{\theta = 1}^T p_\theta = 1$, (i.e., at each time step $k  \geq k'$ one such topology $G_{d_{\theta}}$ is selected independently in an i.i.d. manner). 
Furthermore, let us define the \textit{virtual union digraph} $\mathcal{G}_d^\prime = (\mathcal{V}_{\mathcal{R}}, \cup_{i=1}^{T} \mathcal{E}_{d_i})$. 
The virtual union digraph $\mathcal{G}_d^\prime$ is strongly connected. 
For $k \geq k'$, the union digraph defined as $\mathcal{G}^{ \{ 1, 2, ..., T \}}_d = ( \mathcal{V}_{\mathcal{R}}, \cup_{k= k' + \beta}^{k' + \beta+T-1} \mathcal{E}[k] )$, where $\beta \in \mathbb{Z}_0$, is equal to the virtual union digraph $\mathcal{G}_d^\prime$ which is strongly connected. 
This property means that the open dynamic network $\mathcal{G}_d[k]$ is $T$-\textit{jointly strongly connected} for $k \geq k'$.  
\end{assumption}

Assumption~\ref{awareness_of_remaining_out_neighbors} is important for guaranteeing that each node will perform a transmission towards a node that is active and is participating in the operation of the algorithm, and not to a node that is inactive or departing (i.e., it will be inactive in the next time step). 
Furthermore, it also guarantees that if one node decides to depart (i.e., become inactive at the next time step), it can transmit its stored information to an active node so that the information is not lost after its departure. 
In a network that exhibits directed communication, knowledge of $\mathcal{N}^{+}_j[k] \cap \mathcal{R}[k]$ is challenging but there are ways in which this might be possible. 
One potential approach involves the use of a ``distress signal'' - a special tone transmitted in a control slot or separate channel. 
This signal is sent at higher power than normal communications, enabling it to reach transmitters in its vicinity \cite{2000:bambos_channel}. 
Nodes can gain knowledge of their out-degree by conducting periodic checks. 
These checks might involve transmitting the aforementioned distress signals at regular intervals to identify and count active out-neighbors. 
Such mechanisms allow nodes to maintain awareness of their network connections, even in a dynamic, directed communication environment. 
Another method employs acknowledgment messages, which are common in protocols like TCP, ARQ/HARQ, and ALOHA (see \cite{2024:CDC_Themis_Open}). 
These messages serve to confirm the receipt of information and help overcome channel errors. 
They play a crucial role in ensuring reliable transmissions over unreliable channels. 
Typically, acknowledgment messages are narrowband signals sent via feedback channels. 
This characteristic means they can coexist with the directional data channel without necessarily causing interference.


Assumption~\ref{existence_stable_time_step} ensures that the average of the initial states of active nodes will not change for time steps $k \geq k'$ and the nodes will be able to calculate it by executing our proposed algorithm. 
Note that this is the most common case for analyzing open networks, as they are typically assumed to become eventually closed~\cite[Remark~2]{deplano2025optimization}.
However, our algorithm can also be adapted to settings where the network remains perpetually dynamic by employing an mechanism for departing nodes that relies on memory-based consensus strategies.

Assumption~\ref{strong_connectivity_stable_union_graph} guarantees the existence of at least one directed path between any pair of nodes infinitely often for time steps $k \geq k'$. 
This assumption ensures that information will propagate from each active node towards each other active node in the network 
$\mathcal{G}_d[k] = (\mathcal{V}_{\mathcal{R}}, \mathcal{E}[k])$. 

\begin{remark}
    Current algorithms in the literature operating in OMAS with directed communication links require the network to be strongly connected at all time steps $k \geq 0$ (see \cite{2024:CDC_Hadjic_Garcia, 2024:CDC_Themis_Open}). 
    It is important to note that our paper relaxes this assumption. 
    Specifically, our open dynamic network $\mathcal{G}_d[k]$ does not need to be strongly connected at all time steps $k$. 
    We only require its union digraph $\mathcal{G}^{ \{ 1, 2, ..., T \}}_d$ to be $T$-jointly strongly connected for time steps $k \geq k'$, in order to ensure convergence of our proposed distributed algorithm (see Assumption~\ref{strong_connectivity_stable_union_graph}). 
    This relaxation broadens the applicability of our algorithm to more realistic scenarios where network connectivity may be intermittent or time-varying, while still guaranteeing convergence under weaker connectivity conditions.  
\end{remark}

\subsection{Distributed Quantized Averaging Algorithm in Open Dynamic Networks}\label{subsec_distr_alg_open_dynamic}


Let us consider an open dynamic network $\mathcal{G}_d[k] = (\mathcal{V}[k], \mathcal{E}[k])$. 
Each node $v_j \in \mathcal{V}^\prime$ has an initial quantized state $x_j \in \mathbb{Z}$. 
Additionally, at time step $k$, each node $v_j$ maintains the mass variables $y_j[k]$, $z_j[k]$, the state variables $y_j^s[k]$, $z_j^s[k]$, $q_j^s[k]$, and the transmission variables $c_{lj}^y[k]$, $c_{lj}^z[k]$ at each time step $k$. 
The mass variables are utilized to perform computations on the stored information, the state variables are utilized to store the received information and calculate \eqref{goal}, and the transmission variables are utilized to transmit messages to other nodes.  
During the operation of our algorithm, each node executes the following operations. 


\textbf{Remaining Strategy.} 
During every time step $k$, each active node with remaining status $v_j \in \mathcal{R}[k]$ assigns to each of its outgoing edges $m_{lj}$ (including a virtual self-edge) a nonzero probability  
\begin{equation}\label{eq:remain_trans_prob}
     b_{lj}[k] \hspace{-.1cm} = \hspace{-.1cm} \left\{ 
\begin{aligned} 
\frac{1}{1 + |\mathcal{N}^{+}_j[k] \cap \mathcal{R}[k]|} , \ \hspace{-.05cm} & v_l \in (\mathcal{N}^{+}_j[k] \cap \mathcal{R}[k]) \cup \{v_j\}, \\
0, & \ \text{otherwise}.  
\end{aligned}
\right.
\end{equation}
Then, it updates its state variables to be equal to the mass variables and also updates its transmission variables. 
For updating its transmission variables $c_{lj}^y[k]$, $c_{lj}^z[k]$, it splits $y_j[k]$ into $z_j[k]$ equal pieces, keeps one piece for itself and transmits the other pieces to randomly chosen out-neighbors or itself according to the assigned nonzero probabilities $b_{lj}[k]$ above. 
Then, $v_j$ receives the transmission variables from every in-neighbor $v_i \in \mathcal{N}^{-}_j[k]$, and updates its mass variables $y_j[k+1]$, $z_j[k+1]$ as follows:
\begin{equation}\label{eq:mass_var_update}
\begin{aligned}
    y_j[k+1] = & c_{jj}^y[k] + \sum\limits_{v_i \in \mathcal{N}^{-}_j[k]} w_{ji}[k] \ c_{ji}^y[k] \\
    z_j[k+1] = & c_{jj}^z[k] + \sum\limits_{v_i \in \mathcal{N}^{-}_j[k]} w_{ji}[k] \ c_{ji}^z[k]
    \end{aligned}
\end{equation}
where $w_{ji}=1$ if node $v_j$ receives $c_{ji}^y[k]$, $c_{ji}^z[k]$ from $v_i \in \mathcal{N}^{-}_j[k]$ at iteration $k$ (otherwise $w_{ji}[k] = 0$). 
Note that a more detailed description of the operation performed by the remaining nodes $v_j \in \mathcal{R}[k]$ is shown in \cite[Algorithm~$1$]{2022:Rikos_Hadj_Johan}.



\textbf{Arriving Strategy.} 
When node $v_j$ arrives in the network at time step $k$ (i.e., $v_j \in \mathcal{A}[k]$), it simply updates its state $y_j^s[k+1]$, $z_j^s[k+1]$, $q_j^s[k+1]$, and mass variables $y_j[k+1]$, $z_j[k+1]$, as follows:
\begin{equation}\label{init_variables}
\begin{aligned}
    y_j[k+1] &= 2x_j, \ z_j[k+1] = 2r_j, \\
    y_j^s[k+1] &= 2x_j, \ z_j^s[k+1] = 2r_j, \\
    q_j^s[k+1] &= \left\lfloor \frac{y_j^s[k+1]}{z_j^s[k+1]} \right\rfloor , 
\end{aligned}
\end{equation}
where $r_j = 1$. 
Then, the node $v_j$ starts interacting with its active neighbors at the next time step. 

\textbf{Departing Strategy.} 
When node $v_j$ departs from the network at time step $k$ (i.e., $v_j \in \mathcal{D}[k]$), it assigns to each of its outgoing edges $m_{lj}$ where $v_l \in \mathcal{N}^{+}_j[k] \cap \mathcal{R}[k]$, a nonzero probability value $b_{lj}[k]$ as follows: 
\begin{equation}\label{eq:depart_trans_prob}
     b_{lj}[k]=\left\{ 
\begin{aligned} 
\frac{1}{|\mathcal{N}^{+}_j[k] \cap \mathcal{R}[k]|} , & \ v_l \in \mathcal{N}^{+}_j[k] \cap \mathcal{R}[k], \\
0, & \ v_l \notin \mathcal{N}^{+}_j[k] \cap \mathcal{R}[k]. 
\end{aligned}
\right.
\end{equation}
Then it randomly selects a \textit{remaining} out-neighbor $v_l \in \mathcal{N}^{+}_j[k] \cap \mathcal{R}[k]$ according $b_{lj}[k]$ and transmits towards $v_l$ its negative two times initial state $-2x_j$ and $-2r_j$ combined with its mass variables $y_j[k]$ and $z_j[k]$. 
Specifically, the computation of transmission variables $c_{lj}^y[k]$, $c_{lj}^z[k]$ is as follows:
\begin{equation}\label{trans_vari}
\begin{aligned}
    &c_{lj}^y[k] = y_j[k] - 2x_j,\\
    &c_{lj}^z[k] = z_j[k] - 2r_j. 
\end{aligned}
\end{equation}


In Algorithm~\ref{algorithm1} below we provide a summary of a single step of our proposed procedure. 
This overview outlines all operational modes for each node $v_j$. 

\begin{varalgorithm}{1}
\caption{Quantized Averaging in Open Dynamic Networks}
\noindent \textbf{Input.} 
An open dynamic network $\mathcal{G}_d[k]=(\mathcal{V}[k], \mathcal{E}[k])$ with $n = | \mathcal{V}^\prime |$ nodes potentially participating, $n[k] = | \mathcal{V}[k] |$ active nodes, and $m[k] = | \mathcal{E}[k] |$ edges at each time step $k$. 
Each potentially participating node $v_j \in \mathcal{V}$ has initial state $x_j \in \mathbb{Z}$, and $r_j = 1$. 
Also, Assumptions~\ref{awareness_of_remaining_out_neighbors}, \ref{existence_stable_time_step}, \ref{strong_connectivity_stable_union_graph} hold. 
\\
\textbf{Initialization.} 
Each node $v_j \in \mathcal{V}$ sets $y_j[0] = 2x_j$, $z_j[0] = 2r_j$. \\ 
\textbf{Iteration.} For each time step $k= 0, 1, 2, \dots$\\
\textbf{Arriving:} Each node $v_j \in \mathcal{A}[k]$ updates its state $y_j^s[k+1]$, $z_j^s[k+1]$, $q_j^s[k+1]$, and mass variables $y_j[k+1]$, $z_j[k+1]$ as in \eqref{init_variables}. \\
\textbf{Departing:} Each node $v_j \in \mathcal{D}[k]$ does:
\begin{list4}
\item[$1)$] Assigns a nonzero probability $b_{lj}[k]$ to each of its outgoing edges $m_{lj}$ 
as in \eqref{eq:depart_trans_prob}. 
\item[$2)$] Chooses randomly one remaining out-neighbor $v_l \in \mathcal{N}^{+}_j[k] \cap \mathcal{R}[k]$ according to probability $b_{lj}[k]$. 
\item[$3)$] Computes transmission variables $c_{lj}^y[k]$, $c_{lj}^z[k]$ as in \eqref{trans_vari}. 
\item[$4)$] Transmits $c_{lj}^y[k]$, $c_{lj}^z[k]$ to the selected remaining out-neighbor $v_l$. 
\end{list4} 
\textbf{Remaining:} Each node $v_j \in \mathcal{R}[k]$ does:
\begin{list4} 
\item[$1)$] Assigns to each of its outgoing edges $m_{lj}$ (including a virtual self-edge) a nonzero probability $b_{lj}[k]$ as in $\eqref{eq:remain_trans_prob}$. 
\item[$2)$] Sets $y_j^s[k] = y_j[k], z_j^s[k] = z_j[k], q_j^s[k] = \lfloor\frac{y_j^s[k]}{z_j^s[k]}\rfloor$. 
\item[$3)$] Sets $c_{lj}^y[k] = 0$, $c_{lj}^z[k] = 0$, for every $v_l \in (\mathcal{N}^{+}_j[k] \cap \mathcal{R}[k]) \cup \{v_j\}$. 
\item[$4)$] Sets $\delta_z = z_j[k]$. 
\item[$5)$] \textbf{If} $\delta_z \leq 1$ \textbf{set} $c_{jj}^y[k] = y_j[k]$, $c_{jj}^z[k] = z_j[k]$. 
\item[$6)$] \textbf{While} $\delta_z > 1$ \textbf{do} 
\begin{list4a}
\item[$6.1)$] $\delta_y = \lfloor y[k] / z[k] \rfloor$.
\item[$6.2)$] Chooses $v_l \in (\mathcal{N}^{+}_j[k] \cap \mathcal{R}[k]) \cup \{v_j\}$ according to $b_{lj}[k]$.  
\item[$6.3)$] Sets $c_{lj}^y[k] = c_{lj}^y[k] + \delta_y$, and $c_{lj}^z[k] = c_{lj}^z[k] + 1$ for chosen $v_l$ in previous step.   
\item[$6.4)$] Sets $y_j[k] = y_j[k] - \delta_y$, $z_j[k] = z_j[k] - 1$, $\delta_z = \delta_z - 1$. 
\end{list4a}
\item[$7)$] Transmits $c_{lj}^y[k]$ and $c_{lj}^z[k]$, to $v_l$ for every $v_l \in (\mathcal{N}^{+}_j[k] \cap \mathcal{R}[k]) \cup \{v_j\}$. 
\item[$8)$] Receives $c_{lj}^y[k]$, $c_{lj}^z[k]$ from in-neighbor $v_i \in \mathcal{N}^{-}_j[k]$ and updates $y_j[k+1]$, $z_j[k+1]$ as in \eqref{eq:mass_var_update}. 
\end{list4}
\textbf{Output:} \eqref{alpha_q_no_oscill} holds for every active node $v_j \in \mathcal{V}[k]$. 
\label{algorithm1}
\end{varalgorithm}

\textbf{Intuition.} 
The operation of Algorithm \ref{algorithm1}, at any time step $k$, can be interpreted as the "random walk" of $n[k]$ "tokens" in a \textit{dynamic} (inhomogeneous) Markov chain (i.e., interconnections change over time) with $n[k]=|\mathcal{V}[k]|$ states. 
Note that a detailed analysis of the operation of the underlying averaging algorithm is presented in \cite[Theorem~$1$]{2022:Rikos_Hadj_Johan}. 
Each active node $v_j$ holds two ``tokens'': $T_j^{ins}$ (which is stationary) and $T_j^{out}$ (which performs a random walk).
They each contain a pair of values $y_j^{ins}[k]$, $z_j^{ins}[k]$, and $y_j^{out}[k]$, $z_j^{out}[k]$, respectively, for which it holds that $y_j^{ins}[0] = y_j^{out}[0] = x_j \in\mathbb{Z}$ and $z_j^{ins}[0] = z_j^{out}[0] = r_j=1$. 
The sum of $y^{ins}_j[k],y^{out}_j[k]$ of active nodes is equal to $2x_j$,  and the sum of $z^{ins}_j[k],z^{out}_j[k]$ of active nodes is equal to $2n[k]$ at each time step $k$.
The operations of each mode can then be interpreted as follows. 
 \\ \noindent 
\textit{Remaining.} At each time step $k$, each node $v_j$ keeps the token $T_j^{ins}$ (i.e., it never transmits it) and transmits the token $T_j^{out}$.
When $v_j$ receives one or more tokens $T_i^{out}$ from its in-neighbors, the values $y_i^{out}[k]$ and $y_j^{ins}[k]$ become equal (or differ by at most~$1$). 
Then, $v_j$ transmits each received token $T_i^{out}$ to a randomly selected out-neighbor or itself. 
 \\ \noindent 
\textit{Arriving.} Each arriving node $v_j \in \mathcal{A}[k]$ at time step $k$ also holds two ``tokens'' $T^{ins}_j$ (which is stationary) and $T^{out}_j$ (which performs a random walk). 
These tokens contain pairs of values $y^{ins}_j[k+1]$, $z^{ins}_j[k+1]$, and $y^{out}_j[k+1]$, $z^{out}_j[k+1]$, respectively, for which it holds that $y^{ins}_j[k+1] = y^{out}_j[k+1] = x_j$ and $z^{ins}_j[k+1] = z^{out}_j[k+1] = r_j$. 
 \\ \noindent 
\textit{Departing.} When each departing node $v_j\in\mathcal{D}[k]$ leaves the network, the initial tokens of the departing nodes should be removed, and the initial tokens of the remaining nodes should be preserved. 
Therefore, each departing node transmits its negative values of tokens of time step $k=0$ (i.e., $-(y^{out}_j[0] + y^{ins}_j[0] ) = -2 x_j $, and $-(z^{out}_j[0] + z^{ins}_j[0]) = -2 r_j $) for eliminating the value of its initial tokens.
It also transmits the values of tokens $T^{out}_j$, $T^{ins}_j$ it holds at the time step $k$, (i.e., transmits $y^{out}_j[k] + y^{ins}_j[k] = y_j[k]$, and $z^{out}_j[k] + z^{ins}_j[k] = z_j[k]$), for avoiding losing the values of other remaining nodes' initial tokens.

\subsection{Correctness Analysis of Algorithm~\ref{algorithm1}}

We now establish the correctness of our Algorithm~\ref{algorithm1} via the following theorem. 
Additionally, in our theorem we also provide a necessary and sufficient condition for nodes executing Algorithm~\ref{algorithm1} in an open dynamic network to solve problem \textbf{P1}. 

\begin{theorem}\label{main_convergence_condition_theorem}
Let us consider an open dynamic network $\mathcal{G}_d[k]=(\mathcal{V}[k], \mathcal{E}[k])$ with $n = | \mathcal{V}^\prime |$ nodes potentially participating, $n[k] = | \mathcal{V}[k] |$ active nodes, and $m[k] = | \mathcal{E}[k] |$ edges at each time step $k$.  
Assumptions~\ref{awareness_of_remaining_out_neighbors}, \ref{existence_stable_time_step}, \ref{strong_connectivity_stable_union_graph} hold. 
Let us also assume that all potentially participating nodes execute Algorithm~\ref{algorithm1}. 
There exists a time step $k_0$ such that for every active node $v_j \in \mathcal{V}[k]$ we have 
$$ 
( q^s_j[k] = \lfloor q[k] \rfloor \ \ \text{for} \ \ k \geq k_0 ) \ \text{or} \ ( q^s_j[k] = \lceil q[k] \rceil \ \ \text{for} \ \ k \geq k_0) , 
$$
where $q[k]$ is defined in \eqref{goal}, if and only if for every departing node $v_j \in \mathcal{D}[k]$ it holds that 
\begin{equation}\label{condition_for_correctness} 
    | \ \mathcal{N}^{+}_j[k] \cap \mathcal{R}[k] \ | \geq 1 \ , 
\end{equation} 
for every time step $k$. 
\end{theorem}

\begin{proof}
The structure of our proof is the following. 
Initially, we will analyze the operation of the underlying quantized averaging algorithm (borrowing some notation from \cite[Theorem~$1$]{2022:Rikos_Hadj_Johan}). 
Then, we will focus on the operation of Algorithm~\ref{algorithm1} for time steps $k < k'$ (where $k'$ was defined in Assumption \ref{existence_stable_time_step}) showing how the summation of the initial states of the active nodes is preserved during possible departure and arrival operations from every potentially participating node. 
Finally, for time steps $k \geq k'$, we will show how Algorithm~\ref{algorithm1} enables active nodes to address problem \textbf{P1} in Section~\ref{sec:probForm} (i.e., active nodes are able to calculate the quantized average of their initial states in finite time while exchanging quantized valued messages). 

The operation of Algorithm~\ref{algorithm1}, at any time step $k$, can be interpreted as the ``random walk'' of $n$ ``tokens'' in a \textit{dynamic} (inhomogeneous) Markov chain (i.e., interconnections change over time) with $n[k] = | \mathcal{V}[k] |$ states (i.e., the states of the Markov chain are equal to the number of active nodes at time step $k$). 
Each active node $v_j$ at time step $k=0$ holds two ``tokens'': $T_j^{ins}$ (which is stationary) and $T_j^{out}$ (which performs a random walk).
They each contain a pair of values $y_j^{ins}[k]$, $z_j^{ins}[k]$, and $y_j^{out}[k]$, $z_j^{out}[k]$, respectively, for which it holds that $y_j^{ins}[0] = y_j^{out}[0] = x_j \in \mathbb{Z}$ and $z_j^{ins}[0] = z_j^{out}[0] = r_j = 1$. 
At each time step $k$, each node $v_j$ keeps the token $T_j^{ins}$ (i.e., it never transmits it) and transmits the token $T_j^{out}$, according to the nonzero probability $b_{lj}[k]$ it assigned to its outgoing edges $m_{lj}$. 
When $v_j$ receives one or more tokens $T_i^{out}$ from its in-neighbors, the values $y_i^{out}[k]$ and $y_j^{ins}[k]$ become equal (or differ by at most~$1$).
Then, $v_j$ transmits each received token $T_i^{out}$ to a randomly selected out-neighbor or itself according to the nonzero probability $b_{lj}[k]$. 
Note here that during the operation of Algorithm~\ref{algorithm1}, for every time step $k$ we have 
\begin{equation}\label{sum_preserve}
\sum_{v_j \in \mathcal{V}[k]} y^{out}_j[k] + \sum_{v_j \in \mathcal{V}[k]} y^{ins}_j[k] = 2 \sum_{v_j \in \mathcal{V}[k]} x_j ,
\end{equation}
(i.e., at any given $k$ the sum of all token values in the network among active nodes remains constant, and is equal to twice the initial sum). 

Next, we analyze the effects of arriving and departing. 
We show that \eqref{sum_preserve} holds during every time step $k < k'$ regardless of nodes departing or arriving in the network. 
Note that our analysis focuses on time steps $k < k'$ since for time steps $k > k'$ the subset of active nodes stabilizes (see Assumption \ref{existence_stable_time_step}) and no further arrivals or departures occur. 

\textbf{Arriving:} 
When a node $v_j$ arrives at time step $k$, it initializes its state and mass variables according to \eqref{init_variables}.
Specifically, $v_j$ sets $y_j[k+1] = 2x_j$ and $z_j[k+1] = 2r_j$, and begins interacting with neighbors in the subsequent time step.
Each arriving node $v_j \in \mathcal{A}[k]$ at time step $k$ also holds two ``tokens'' $T^{ins}_j$ (which is stationary) and $T^{out}_j$ (which performs a random walk). 
These tokens contain pairs of values $y^{ins}_j[k+1]$, $z^{ins}_j[k+1]$, and $y^{out}_j[k+1]$, $z^{out}_j[k+1]$, respectively, for which it holds that $y^{ins}_j[k+1] = y^{out}_j[k+1] = x_j \in \mathbb{Z}$ and $z^{ins}_j[k] = z^{out}_j[k] = r_j = 1$. 
This means that at time step $k$ we have
\begin{equation}\label{eq:equation_arri}
    \sum\limits_{v_j \in \mathcal{A}[k]} y^{out}_j[k+1] + \hspace{-.25cm} \sum\limits_{v_j \in \mathcal{A}[k]} y^{ins}_j[k+1] = 2 \hspace{-.25cm} \sum\limits_{v_j \in \mathcal{A}[k]}x_j .
\end{equation} 
By combining \eqref{eq:equation_arri} with \eqref{sum_preserve}, we can conclude that \eqref{sum_preserve} continues to hold at time step $k+1$, maintaining the sum preservation property even as new nodes join the network.

\textbf{Departing:} 
Let us suppose that node $v_j$ departs from the network at time step $k$ and it has at least one remaining out-neighbor (i.e., $| \ \mathcal{N}^{+}_j[k] \cap \mathcal{R}[k] \ | \geq 1$). 
Node $v_j$ assigns nonzero probabilities $b_{lj}[k]$ as in \eqref{eq:depart_trans_prob}. 
It then computes the transmission variables as in \eqref{trans_vari}, and transmits them to one randomly chosen remaining out-neighbor. 
Following this transmission, $v_j$ exits the network and becomes inactive in the subsequent time step. 
It is important to note that if $v_j$ has no remaining out-neighbors (i.e., $| \ \mathcal{N}^{+}_j[k] \cap \mathcal{R}[k] \ | = 0$), it cannot compute $b_{lj}[k]$ or transmit to a remaining out-neighbor. 
In this case, when $v_j$ leaves the network the information it holds will be lost. 
This lost information might include crucial data about the initial states of other remaining nodes, which could be essential for solving problem \textbf{P1}. 
Note now that if the tokens of the departing nodes are removed from the network then we have  
\begin{equation}\label{eq:equation_depar}
\begin{aligned}
&\sum\limits_{v_j \in \mathcal{V}[k]} y^{out}_j[k] + \sum\limits_{v_j \in \mathcal{V}[k]} y^{ins}_j[k] \\ & -(\sum\limits_{v_j \in \mathcal{D}[k]} y^{out}_j[0] + \sum\limits_{v_j \in \mathcal{D}[k]} y^{ins}_j[0]) \\ & = 2(\sum\limits_{v_j \in \mathcal{V}[k]} x_j - \sum\limits_{v_j \in \mathcal{D}[k]} x_j), \ \forall \ k < k',
\end{aligned} 
\end{equation} 
where $y^{ins}_j[0] = y^{out}_j[0] = x_j \in \mathbb{Z}$. 
Before leaving the network, each departing node $v_j \in \mathcal{D}[k]$ transmits its negative values of tokens of time step $k=0$ (i.e., $-(y^{out}_j[0] + y^{ins}_j[0])$, $-(z^{out}_j[0] + z^{ins}_j[0])$, where $-(y^{out}_j[0] + y^{ins}_j[0] ) = -2 x_j = -y_j[0]$, and $-(z^{out}_j[0] + z^{ins}_j[0]) = -2 r_j = -z_j[0]$) for eliminating the departing nodes' information. 
In addition, each departing node $v_j$ also transmits the values of tokens $T^{out}_j$, $T^{ins}_j$ it holds at the time step $k$, (i.e., transmits $y^{out}_j[k] + y^{ins}_j[k]$, and $z^{out}_j[k] + z^{ins}_j[k]$, where $y^{out}_j[k] + y^{ins}_j[k] = y_j[k]$, and $z^{out}_j[k] + z^{ins}_j[k] = z_j[k]$), for avoiding losing the other remaining nodes' information. 
Therefore, by transmitting the transmission variables defined in \eqref{trans_vari} then \eqref{eq:equation_depar} holds true. 
We also have that \eqref{eq:equation_depar} is equivalent to 
\begin{equation}\label{eq:equation_depar_equa} 
    \sum\limits_{v_j \in \mathcal{R}[k]} y^{out}_j[k] + \sum\limits_{v_j \in  \mathcal{R}[k]} y^{ins}_j[k] = 2 \sum\limits_{v_j \in  \mathcal{R}[k]} x_j, 
\end{equation}
for $k < k'$. 
Let us now assume for simplicity that no arrivals occur at time step $k$ (the case where arrivals and departures occur simultaneously at time step $k$ can be analyzed by combining equations \eqref{eq:equation_arri} and \eqref{eq:equation_depar_equa}). 
This means that \eqref{eq:equation_depar_equa} becomes equivalent to 
\begin{equation}\label{eq:equation_k+1}
\begin{aligned} 
&\sum\limits_{v_j \in \mathcal{V}[k+1]} y^{out}_j[k+1] + \sum\limits_{v_j \in \mathcal{V}[k+1]} y^{ins}_j[k+1] \\ & = 2\sum\limits_{v_j \in \mathcal{V}[k+1]} x_j, \ \forall \ k < k'.  
\end{aligned}
\end{equation}
As a result, \eqref{sum_preserve} remains valid at time step $k+1$. 
This means that the sum preservation property is maintained even when nodes leave the network. 

Let us now examine the network behavior for time steps $k \geq k'$. According to Assumption~\ref{existence_stable_time_step}, no further node arrivals or departures occur after this point. 
However, the communication links between nodes continue to change over time. 
This results in an open dynamic network defined as $\mathcal{G}_d[k] = (\mathcal{V}_{\mathcal{R}}, \mathcal{E}[k])$ for $k \geq k'$. 
For $\mathcal{G}_d[k]$ we have that Assumption~\ref{strong_connectivity_stable_union_graph} holds. 
This means that for time steps $k \geq k'$ the open dynamic network $\mathcal{G}_d[k]$ is $T$-jointly strongly connected (i.e., its virtual union digraph $\mathcal{G}_d^\prime$ is strongly connected).  
Given that $\mathcal{G}_d[k]$ is $T$-jointly strongly connected and experiences no node arrivals or departures for $k \geq k'$, we can analyze its convergence in a manner similar to that presented in \cite[Theorem~$1$]{2022:Rikos_Hadj_Johan}. 
We omit the detailed proof here due to space limitations, but the analysis follows the same principles as in the cited theorem.
\end{proof}

\begin{remark}[Intuition of Theorem~\ref{main_convergence_condition_theorem}]
    In Theorem~\ref{main_convergence_condition_theorem} we show that departing nodes must have at least one remaining out-neighbor to ensure that Algorithm~\ref{algorithm1} enables active nodes to solve problem \textbf{P1}. 
    If a departing node can transmit to at least one remaining node before departure it successfully eliminates its initial state from the network and also preserves essential information. 
    This condition allows Algorithm~\ref{algorithm1} to converge correctly as it preserves useful information and eliminates outdated data from the network. 
    However, when a departing node has no remaining out-neighbors then its stored information is lost. 
    This may lead to lead to inconsistency in the network’s state and prevent the active nodes from correctly solving problem~\textbf{P1}. 
\end{remark}
%
%
%
%
\section{SIMULATION RESULTS}\label{sec:simulation}
We now illustrate the performance of Algorithm~\ref{algorithm1} with numerical simulations. 
We consider an open dynamic network $\mathcal{G}_d[k] = (\mathcal{V}[k], \mathcal{E}[k])$ with $n = 150$ potentially participating nodes. 
Assumptions~\ref{awareness_of_remaining_out_neighbors}, \ref{existence_stable_time_step}, \ref{strong_connectivity_stable_union_graph} hold. 
At time step $k=0$ we have $n[0] = 100$ active nodes. 
Each initially active node $v_j$ has an initial state $x_j$ that is an integer randomly chosen from the interval $[1, 10]$ with uniform probability. 
The nodes that are inactive at time step $k=0$ (i.e., the nodes that belong in the set $\mathcal{V}^\prime \setminus \mathcal{V}[0]$) and nodes that may enter the network at time step $k\geq1$ (i.e., nodes that were inactive at $k=0$ or nodes that departed at $k\geq1$ and entered again at a subsequent time step) have initial states that are integers randomly chosen from the interval $[10, 20]$ with uniform probability.
Between the time intervals $1 < k \leq 80$ and $150 < k \leq 230$, the network size may increase or decrease by 1 due to agent arrival or departure, with $10\%$ in the former interval, and $20\%$ in the latter, which is followed \cite[Section~V]{makridis2024average}.
Between the time intervals $80 < k \leq 150$ and $230 < k \leq 300$, the network is assumed to be stable (no arriving or departure). 
Additionally, we have (i) $k' = 230$ (see Assumption~\ref{existence_stable_time_step}), (ii) our network $\mathcal{G}_d[k]$ is $T$-jointly strongly connected for $T = 20$ and (see Assumption~\ref{strong_connectivity_stable_union_graph}), and (iii) for every departing node $v_j \in \mathcal{D}[k]$ it holds that  $| \ \mathcal{N}^{+}_j[k] \cap \mathcal{R}[k] \ | \geq 1$ during every time step $k$ (see Theorem~\ref{main_convergence_condition_theorem}). 
In what follows we illustrate over time $k$ (i) the evolution of the states of active nodes $q^s_j[k]$, and (ii) the time-varying average consensus error $\varepsilon[k]$ defined as
\begin{equation}
\begin{aligned}
\varepsilon[k] = & \sum\limits_{\{v_j \in \mathcal{V}[k] | \lceil y_j[k] / z_j[k] \rceil > \lceil q[k] \rceil\}} \!\!\!\!\!\!\!\!\!\!\!\!\!(\lceil 
 y_j[k] / z_j[k] \rceil \!-\! \lceil q[k] \rceil) +\! \\&\sum\limits_{\{v_j \in \mathcal{V}[k] | \lfloor y_j[k] / z_j[k] \rfloor < \lfloor q[k] \rfloor\}}\!\!\!\!\!\!\!\!\!\!\!\!\!(\lfloor q[k] \rfloor - \lfloor y_j[k] / z_j[k] \rfloor) , 
\end{aligned}\label{error_plot}
\end{equation}
where $q[k]$ is defined in \eqref{goal}. 


In Fig.~\ref{fig:state variables in open dynamic network} we present the evolution of the state variable $q^s_j[k]$ of active agents $v_j \in \mathcal{V}[k]$ during the execution of Algorithm~\ref{algorithm1}. 
Also, the real average of the initial states of the active agents $q[k]$ at time step $k$ is shown in the black dashed line. 
We can see that for time steps $1 < k \leq 80$ the states of the active nodes are tracking the quantized average ($\lfloor q[k] \rfloor$, or $\lceil q[k] \rceil$) but do not stabilize to a specific value due to nodes arriving and departing (which causes the value $q[k]$ to change). 
However, in the time interval $80 < k \leq 150$ the active nodes' states are becoming equal to the quantized average of their initial states $\lfloor q[k] \rfloor$, or $\lceil q[k] \rceil$. 
Following that, for time steps $151 < k \leq 230$ the active nodes again track the quantized average without stabilizing due to continued node arrivals and departures. 
Finally, for $230 < k \leq 300$ no departures or arrivals occur and the active nodes' states are becoming equal to the quantized average of their initial states, thus solving problem \textbf{P1}.

In Fig.~\ref{fig:error in open dynamic network} we present the evolution of $\varepsilon[k]$ (shown in \eqref{error_plot}) over time $k$ during the execution of Algorithm~\ref{algorithm1}. 
This plot provides us similar insights as Fig.~\ref{fig:state variables in open dynamic network}. 
Specifically, for time intervals $1 < k \leq 80$ and $151 < k \leq 230$ where nodes arrive or depart the network, the error $\varepsilon[k]$ increases and fluctuates due to $q[k]$ changing with time. 
However, for time intervals $80 < k \leq 150$ and $230 < k \leq 300$ where no arrivals or departures occur $\varepsilon[k]$ becomes equal to zero in finite time. 
As a result, this plot confirms that that Algorithm~\ref{algorithm1} enables active nodes to successfully solve problem \textbf{P1} when the network stabilizes. 


\begin{figure}
\begin{center}
\includegraphics[width=0.86\columnwidth]{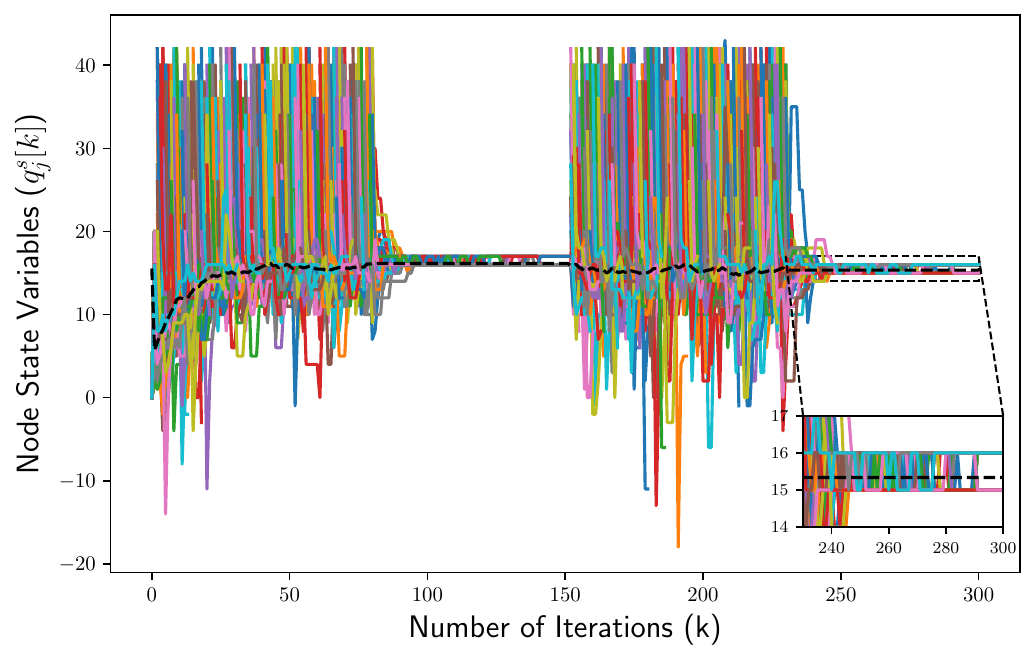}    
\caption{Evolution of active agents' state variables $q^s_j[k]$ over time $k$ during the execution of Algorithm~\ref{algorithm1}.}
\label{fig:state variables in open dynamic network}
\end{center}
\end{figure}


\begin{figure}
\begin{center}
\includegraphics[width=0.85\columnwidth]{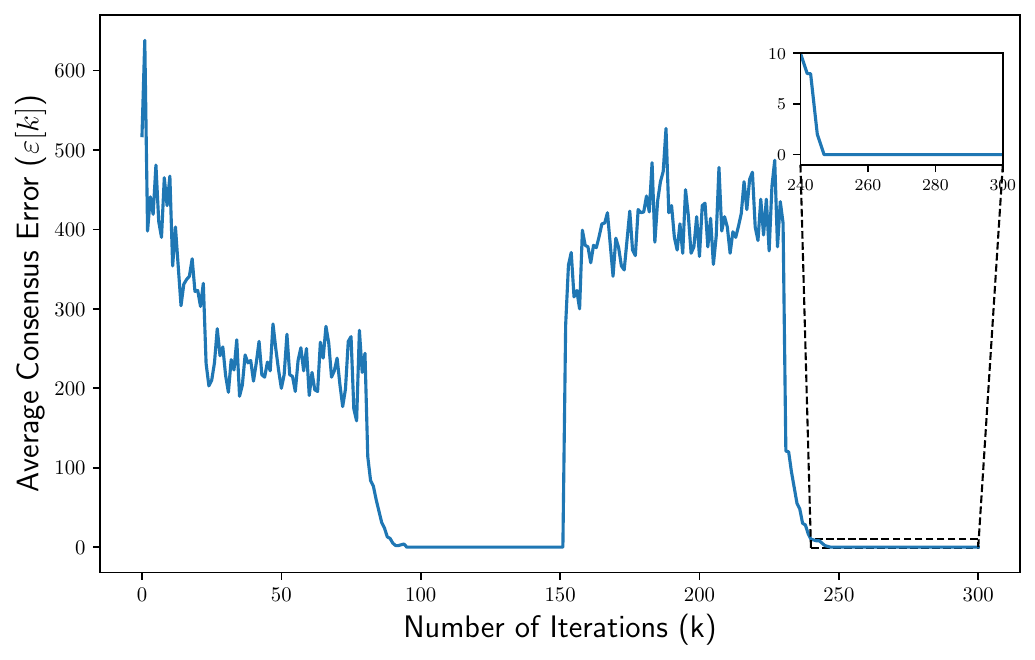} 
\caption{Evolution of the average consensus error $\varepsilon[k]$ (shown in \eqref{error_plot}) over time $k$ during the execution of Algorithm~\ref{algorithm1}.}
\label{fig:error in open dynamic network}
\end{center}
\end{figure}




%
%
%
%
\section{Conclusions and Future Directions}\label{sec:conclusions}

In this paper, we presented a novel distributed average consensus algorithm designed to operate in open dynamic multi-agent systems with directed communication links. 
Our proposed algorithm is the first to enable nodes to (i) exhibit efficient (quantized) communication, (ii) operate in the present of time-varying links among active nodes. 
We analyzed the operation of our algorithm and established its correctness. 
Additionally, we presented a set of topological conditions that enable nodes to calculate the desired result in finite time. 
We concluded our paper with a numerical simulation for illustrating the correctness of our algorithm.  

In the future we plan to design robust algorithms capable of operating in open multi-agent systems subject to various network unreliabilities, such as packet losses.
Additionally, we intend to expand our focus to address the challenges of distributed optimization in OMAS.



\bibliographystyle{IEEEtran}
\bibliography{IEEEabrv,reference}

\begin{thebibliography}{10}
\providecommand{\url}[1]{#1}
\csname url@samestyle\endcsname
\providecommand{\newblock}{\relax}
\providecommand{\bibinfo}[2]{#2}
\providecommand{\BIBentrySTDinterwordspacing}{\spaceskip=0pt\relax}
\providecommand{\BIBentryALTinterwordstretchfactor}{4}
\providecommand{\BIBentryALTinterwordspacing}{\spaceskip=\fontdimen2\font plus
\BIBentryALTinterwordstretchfactor\fontdimen3\font minus \fontdimen4\font\relax}
\providecommand{\BIBforeignlanguage}[2]{{%
\expandafter\ifx\csname l@#1\endcsname\relax
\typeout{** WARNING: IEEEtran.bst: No hyphenation pattern has been}%
\typeout{** loaded for the language `#1'. Using the pattern for}%
\typeout{** the default language instead.}%
\else
\language=\csname l@#1\endcsname
\fi
#2}}
\providecommand{\BIBdecl}{\relax}
\BIBdecl

\bibitem{2007:olfati-saber_consensus}
R.~{Olfati-Saber}, J.~Fax, and R.~Murray, ``Consensus and cooperation in networked multi-agent systems,'' \emph{Proceedings of the {IEEE}}, vol.~95, no.~1, pp. 215--233, Jan. 2007.

\bibitem{2018:BOOK}
C.~N. Hadjicostis, A.~D. Domínguez-García, and T.~Charalambous, ``Distributed averaging and balancing in network systems, with applications to coordination and control,'' \emph{Foundations and Trends\textregistered ~in Systems and Control}, vol.~5, no. 3--4, 2018.

\bibitem{SEYBOTH:2013}
G.~S. Seyboth, D.~V. Dimarogonas, and K.~H. Johansson, ``Event-based broadcasting for multi-agent average consensus,'' \emph{Automatica}, vol.~49, no.~1, pp. 245--252, 2013.

\bibitem{2021:Rikos_Hadj_Accumul_TAC}
A.~I. Rikos and C.~N. Hadjicostis, ``Event-triggered quantized average consensus via ratios of accumulated values,'' \emph{IEEE Transactions on Automatic Control}, vol.~66, no.~3, pp. 1293--1300, 2021.

\bibitem{chrisTAC:2016}
C.~N. Hadjicostis, N.~H. Vaidya, and A.~D. Domínguez-García, ``Robust distributed average consensus via exchange of running sums,'' \emph{IEEE Transactions on Automatic Control}, vol.~61, no.~6, pp. 1492--1507, June 2016.

\bibitem{Vizuete2024}
R.~Vizuete, C.~Monnoyer~de Galland, P.~Frasca, E.~Panteley, and J.~M. Hendrickx, \emph{Trends and Questions in Open Multi-agent Systems}.\hskip 1em plus 0.5em minus 0.4em\relax Springer Nature Switzerland, 2024, pp. 219--252.

\bibitem{huynh2006integrated}
T.~D. Huynh, N.~R. Jennings, and N.~R. Shadbolt, ``An integrated trust and reputation model for open multi-agent systems,'' \emph{Autonomous agents and multi-agent systems}, vol.~13, pp. 119--154, 2006.

\bibitem{2017_Hendrickx_Martin_CDC}
J.~M. Hendrickx and S.~Martin, ``Open multi-agent systems: Gossiping with random arrivals and departures,'' in \emph{Proceedings of the IEEE Conference on Decision and Control}, 2017, pp. 763--768.

\bibitem{golpayegani2019using}
F.~Golpayegani, I.~Dusparic, and S.~Clarke, ``Using social dependence to enable neighbourly behaviour in open multi-agent systems,'' \emph{ACM Transactions on Intelligent Systems and Technology (TIST)}, vol.~10, no.~3, pp. 1--31, 2019.

\bibitem{2016_Hendrickx_Allerton}
J.~M. Hendrickx and S.~Martin, ``Open multi-agent systems: Gossiping with deterministic arrivals and departures,'' in \emph{Annual Allerton Conference on Communication, Control, and Computing}, 2016, pp. 1094--1101.

\bibitem{2021_Franceschelli_Frasca_TAC}
M.~Franceschelli and P.~Frasca, ``Stability of open multiagent systems and applications to dynamic consensus,'' \emph{IEEE Transactions on Automatic Control}, vol.~66, no.~5, pp. 2326--2331, 2021.

\bibitem{2022_Dashti_Mauro_IEEELCSS}
Z.~A. Z.~S. Dashti, G.~Oliva, C.~Seatzu, A.~Gasparri, and M.~Franceschelli, ``Distributed mode computation in open multi-agent systems,'' \emph{IEEE Control Systems Letters}, vol.~6, pp. 3481--3486, 2022.

\bibitem{2024_Oliva_Scala_TAC_Open}
G.~Oliva, M.~Franceschelli, A.~Gasparri, and A.~Scala, ``A sum-of-states preservation framework for open multiagent systems with nonlinear heterogeneous coupling,'' \emph{IEEE Transactions on Automatic Control}, vol.~69, no.~3, pp. 1991--1998, 2024.

\bibitem{2020_Hendrickx_Rabbat_CDC}
J.~M. Hendrickx and M.~G. Rabbat, ``Stability of decentralized gradient descent in open multi-agent systems,'' in \emph{Proceedings of the IEEE Conference on Decision and Control}, 2020, pp. 4885--4890.

\bibitem{2023_Hayashi_TAC}
N.~Hayashi, ``Distributed subgradient method in open multiagent systems,'' \emph{IEEE Transactions on Automatic Control}, vol.~68, no.~10, pp. 6192--6199, 2023.

\bibitem{2023_Nakamura_Inuiguchi_IEEECSS}
T.~Nakamura, N.~Hayashi, and M.~Inuiguchi, ``Cooperative learning for adversarial multi-armed bandit on open multi-agent systems,'' \emph{IEEE Control Systems Letters}, vol.~7, pp. 1712--1717, 2023.

\bibitem{2024:CDC_Hadjic_Garcia}
C.~N. Hadjicostis and A.~D. Dominguez-Garcia, ``Distributed average consensus in open multi-agent systems,'' in \emph{Proceedings of the IEEE Conference on Decision and Control}, 2024, pp. 3037--3042.

\bibitem{2024:CDC_Themis_Open}
E.~Makridis, A.~Grammenos, G.~Oliva, E.~Kalyvianaki, C.~N. Hadjicostis, and T.~Charalambous, ``Average consensus over directed networks in open multi-agent systems with acknowledgement feedback,'' in \emph{Proceedings of the IEEE Conference on Decision and Control}, 2024, pp. 3051--3056.

\bibitem{2022:Rikos_Hadj_Johan}
A.~I. Rikos, C.~N. Hadjicostis, and K.~H. Johansson, ``Non-oscillating quantized average consensus over dynamic directed topologies,'' \emph{Automatica}, vol. 146, p. 110621, 2022.

\bibitem{2000:bambos_channel}
N.~Bambos, S.~C. Chen, and G.~J. Pottie, ``Channel access algorithms with active link protection for wireless communication networks with power control,'' \emph{IEEE\slash ACM Transactions on Networking}, vol.~8, no.~5, pp. 583--597, 2000.

\bibitem{deplano2025optimization}
D.~Deplano, N.~Bastianello, M.~Franceschelli, and K.~H. Johansson, ``Optimization and learning in open multi-agent systems,'' \emph{arXiv preprint arXiv:2501.16847}, 2025.

\bibitem{makridis2024average}
E.~Makridis, A.~Grammenos, G.~Oliva, E.~Kalyvianaki, C.~N. Hadjicostis, and T.~Charalambous, ``Average consensus over directed networks in open multi-agent systems with acknowledgement feedback,'' \emph{arXiv preprint arXiv:2409.08634}, 2024.

\end{thebibliography}

\end{document}